\documentclass[1p,final]{elsarticle}
\usepackage{amsfonts,color,morefloats,pslatex}
\usepackage{amssymb,amsthm, amsmath,latexsym}

\newtheorem{theorem}{Theorem}
\newtheorem{lemma}[theorem]{Lemma}

\newtheorem{corollary}[theorem]{Corollary}

\newcommand{\lcm}{{\mathrm{lcm}}}

\newcommand{\gf}{{\mathrm{GF}}}

\newcommand{\AGL}{{\mathrm{AGL}}}

\newcommand{\PAut}{{\mathrm{PAut}}} 
\newcommand{\MAut}{{\mathrm{MAut}}} 
\newcommand{\GAut}{{\mathrm{Aut}}}

\newcommand{\Sym}{{\mathrm{Sym}}}

\newcommand{\wt}{{\mathtt{wt}}}

\newcommand{\m}{\mathbb{M}}

\newcommand{\cP}{{\mathcal{P}}} 
\newcommand{\cB}{{\mathcal{B}}}

\newcommand{\C}{{\mathsf{C}}}

\newcommand{\cR}{{\mathcal{R}}}

\newcommand{\bzero}{{\mathbf{\bar{0}}}}
\newcommand{\bone}{{\mathbf{\bar{1}}}}

\newcommand{\bD}{{\mathbb{D}}}

\newcommand{\cJ}{{\mathtt{J}}}
\newcommand{\cN}{{\mathcal{N}}}

\begin{document}

\begin{frontmatter}



\title{An infinite family of Steiner systems $S(2, 4, 2^m)$ from cyclic codes  
\tnotetext[fn1]{C. Ding's research was supported by the Hong Kong Research Grants Council,
Proj. No. 16300415.}
}

\author{Cunsheng Ding}
\ead{cding@ust.hk}

%\cortext[lcj]{Corresponding author}
\address{Department of Computer Science and Engineering, 
The Hong Kong University of Science and Technology,
Clear Water Bay, Kowloon, Hong Kong, China}

\begin{abstract}
Steiner systems are a fascinating topic of combinatorics. The most studied Steiner systems 
are $S(2, 3, v)$ (Steiner triple systems), $S(3, 4, v)$ (Steiner quadruple systems), and 
$S(2, 4, v)$. There are a few infinite families of Steiner systems $S(2, 4, v)$ in the 
literature. The objective of this paper is to present an infinite family of Steiner systems 
$S(2, 4, 2^m)$ for all $m \equiv 2 \pmod{4} \geq 6$ from cyclic codes. This may be the first 
coding-theoretic construction of an infinite family of Steiner systems $S(2, 4, v)$. As a 
by-product, many infinite families of $2$-designs are also reported in this paper.       
\end{abstract}


\end{frontmatter}

\section{Introduction}

Let $\cP$ be a set of $v \ge 1$ elements, and let $\cB$ be a set of $k$-subsets of $\cP$, where $k$ is
a positive integer with $1 \leq k \leq v$. Let $t$ be a positive integer with $t \leq k$. The pair
$\bD = (\cP, \cB)$ is called a $t$-$(v, k, \lambda)$ {\em design\index{design}}, or simply {\em $t$-design\index{$t$-design}}, if every $t$-subset of $\cP$ is contained in exactly $\lambda$ elements of
$\cB$. The elements of $\cP$ are called points, and those of $\cB$ are referred to as blocks.
We usually use $b$ to denote the number of blocks in $\cB$.  A $t$-design is called {\em simple\index{simple}} if $\cB$ does not contain repeated blocks. In this paper, we consider only simple 
$t$-designs.  A $t$-design is called {\em symmetric\index{symmetric design}} if $v = b$. It is clear that $t$-designs with $k = t$ or $k = v$ always exist. Such $t$-designs are {\em trivial}. In this paper, we consider only $t$-designs with $v > k > t$.
A $t$-$(v,k,\lambda)$ design is referred to as a {\em Steiner system\index{Steiner system}} if $t \geq 2$ and $\lambda=1$, and is denoted by $S(t,k, v)$.

One of the interesting topics in $t$-designs is the study of Steiner systems $S(2, 4, v)$. 
It is known that a Steiner system $S(2, 4, v)$ exists if and only if $v \equiv 1 \mbox{ or } 4 \pmod{12}$ \cite{Hanani}. According to the surveys \cite{CMhb,RR10}, the following is a list 
of infinite families of Steiner systems $S(2,4, v)$: 
\begin{itemize}
\item $S(2, 4, 4^n)$, $n \geq 2$ (affine geometries). 
\item $S(2, 4, 3^n+\cdots+3+1)$, $n \geq 2$ (projective geometries). 
\item $S(2, 4, 2^{s+2}-2^s+4)$, $s >2$ (Denniston designs).  
\end{itemize}   
The objective of this paper is to present an infinite family of Steiner systems $S(2, 4, 2^m)$ 
for all $m \equiv 2 \mod{4} \geq 6$ with extended primitive cyclic codes. This may be the first 
coding-theory construction of an infinite family of Steiner systems $S(2, 4, v)$. As a by product, 
this paper will also construct a number of infinite families of $2$-designs with these binary codes.

\section{The classical construction of $t$-designs from codes}

We assume that the reader is familiar with the basics of linear codes and cyclic codes, and proceed to 
introduce the classical construction of $t$-designs from codes directly. 
Let $\C$ be a $[v, \kappa, d]$ linear code over $\gf(q)$. Let $A_i:=A_i(\C)$, which denotes the
number of codewords with Hamming weight $i$ in $\C$, where $0 \leq i \leq v$. The sequence 
$(A_0, A_1, \cdots, A_{v})$ is
called the \textit{weight distribution} of $\C$, and $\sum_{i=0}^v A_iz^i$ is referred to as
the \textit{weight enumerator} of $\C$. For each $k$ with $A_k \neq 0$,  let $\cB_k$ denote
the set of the supports of all codewords with Hamming weight $k$ in $\C$, where the coordinates of a codeword
are indexed by $(0,1,2, \cdots, v-1)$. Let $\cP=\{0, 1, 2, \cdots, v-1\}$.  The pair $(\cP, \cB_k)$
may be a $t$-$(v, k, \lambda)$ design for some positive integer $\lambda$, which is called a 
\emph{support design} of the code. In such a case, we say that the code $\C$ holds a $t$-$(v, k, \lambda)$ 
design. Throughout this paper, we denote the dual code of $\C$ by $\C^\perp$, and the extended code of 
$\C$ by $\overline{\C}$.   

\subsection{Designs from linear codes via the Assmus-Mattson Theorem}

The following theorem, developed by Assumus and Mattson, shows that the pair $(\cP, \cB_k)$ defined by 
a linear code is a $t$-design under certain conditions \cite{AM74}, \cite[p. 303]{HP03}.

\begin{theorem}[Assmus-Mattson Theorem]\label{thm-designAMtheorem}
Let $\C$ be a $[v,k,d]$ code over $\gf(q)$. Let $d^\perp$ denote the minimum distance of $\C^\perp$. 
Let $w$ be the largest integer satisfying $w \leq v$ and 
$$ 
w-\left\lfloor  \frac{w+q-2}{q-1} \right\rfloor <d. 
$$ 
Define $w^\perp$ analogously using $d^\perp$. Let $(A_i)_{i=0}^v$ and $(A_i^\perp)_{i=0}^v$ denote 
the weight distribution of $\C$ and $\C^\perp$, respectively. Fix a positive integer $t$ with $t<d$, and 
let $s$ be the number of $i$ with $A_i^\perp \neq 0$ for $0 \leq i \leq v-t$. Suppose $s \leq d-t$. Then 
\begin{itemize}
\item the codewords of weight $i$ in $\C$ hold a $t$-design provided $A_i \neq 0$ and $d \leq i \leq w$, and 
\item the codewords of weight $i$ in $\C^\perp$ hold a $t$-design provided $A_i^\perp \neq 0$ and 
         $d^\perp \leq i \leq \min\{v-t, w^\perp\}$. 
\end{itemize}
\end{theorem}

The Assmus-Mattson Theorem is a very useful tool in constructing $t$-designs from linear codes, 
and has been recently employed to construct infinitely many $2$-designs and $3$-designs in 
\cite{DingLi16} and \cite{DingDesign16}.

\subsection{Designs from linear codes via the automorphism group}

In this section, we introduce the automorphism approach to obtaining $t$-designs from linear codes. 
To this end, we have to define the automorphism group of linear codes. We will also present some 
basic results about this approach.

The set of coordinate permutations that map a code $\C$ to itself forms a group, which is referred to as 
the \emph{permutation automorphism group\index{permutation automorphism group of codes}} of $\C$
and denoted by $\PAut(\C)$. If $\C$ is a code of length $n$, then $\PAut(\C)$ is a subgroup of the 
\emph{symmetric group\index{symmetric group}} $\Sym_n$.

A \emph{monomial matrix\index{monomial matrix}} over $\gf(q)$ is a square matrix having exactly one 
nonzero element of $\gf(q)$  in each row and column. A monomial matrix $M$ can be written either in 
the form $DP$ or the form $PD_1$, where $D$ and $D_1$ are diagonal matrices and $P$ is a permutation 
matrix. 

The set of monomial matrices that map $\C$ to itself forms the group $\MAut(\C)$,  which is called the 
\emph{monomial automorphism group\index{monomial automorphism group}} of $\C$. Clearly, we have 
$$
\PAut(\C) \subseteq \MAut(\C).
$$

The \textit{automorphism group}\index{automorphism group} of $\C$, denoted by $\GAut(\C)$, is the set 
of maps of the form $M\gamma$, 
where $M$ is a monomial matrix and $\gamma$ is a field automorphism, that map $\C$ to itself. In the binary 
case, $\PAut(\C)$,  $\MAut(\C)$ and $\GAut(\C)$ are the same. If $q$ is a prime, $\MAut(\C)$ and 
$\GAut(\C)$ are identical. In general, we have 
$$ 
\PAut(\C) \subseteq \MAut(\C) \subseteq \GAut(\C). 
$$

By definition, every element in $\GAut(\C)$ is of the form $DP\gamma$, where $D$ is a diagonal matrix, 
$P$ is a permutation matrix, and $\gamma$ is an automorphism of $\gf(q)$.   
The automorphism group $\GAut(\C)$ is said to be $t$-transitive if for every pair of $t$-element ordered 
sets of coordinates, there is an element $DP\gamma$ of the automorphism group $\GAut(\C)$ such that its 
permutation part $P$ sends the first set to the second set.

A proof of the following theorem can be found in \cite[p. 308]{HP03}. 
 
\begin{theorem}\label{thm-designCodeAutm}
Let $\C$ be a linear code of length $n$ over $\gf(q)$ where $\GAut(\C)$ is $t$-transitive. Then the codewords of any weight $i \geq t$ of $\C$ hold a $t$-design.
\end{theorem}

This theorem gives another sufficient condition for a linear code to hold $t$-designs. To apply Theorem 
\ref{thm-designCodeAutm}, we have to determine the automorphism group of $\C$ and show that it is $t$-transitive. 
It is in general very hard to find out the automorphism group of a linear code. Even if we known that a 
linear code holds $t$-$(v, k, \lambda)$ designs, determining the parameters $k$ and $\lambda$ could be 
extremely difficult. All the $2$-designs presented in this paper are obtained 
from this automorphism group approach. 

The next theorem will be employed later and is a very useful and general result  \cite[p. 165]{MS77}. 

\begin{theorem}\label{thm-dualdesignbinary}
Let $\C$ be an $[n, k, d]$ binary linear code with $k>1$, such that for each weight $w>0$ the supports of the codewords of weight $w$ form a $t$-design, where $t<d$. Then the supports of the codewords of each nonzero weight in $\C^\perp$ also form a $t$-design.   
\end{theorem}

\section{Affine-invariant linear codes} 

In this section, we first give a special representation of primitive cyclic codes and their extended codes, 
and then define and characterise affine-invariant codes. We will skip proof details, but 
refer the reader to \cite[Section 4.7]{HP03} for a detailed proof of the major results presented in this section. 

A cyclic code of length $n=q^m-1$ over $\gf(q)$ for some positive integer $m$ is called a \emph{primitive 
cyclic code}. Let $\cR_n$ denote the quotient ring $\gf(q)[x]/(x^n-1)$. Any cyclic code $\C$ of length 
$n=q^m-1$ over $\gf(q)$ is an ideal of $\cR_n$, and is generated by a monic polynomial $g(x)$ of the 
least degree over $\gf(q)$. This polynomial is called the generator polynomial of the cyclic code $\C$, and 
can be expressed as 
$$ 
g(x)=\prod_{t \in T} (x-\alpha^t), 
$$ 
where $\alpha$ is a generator of $\gf(q^m)^*$, $T$ is a subset of $\cN=\{0,1, \cdots, n-1\}$ and a union 
of some $q$-cyclotomic cosets modulo $n$. 
The set $T$ is called a \emph{defining set} of $\C$ with respect to $\alpha$. 
When $\C$ is viewed as a subset of $\cR_n$, 
every codeword of $\C$ is a polynomial 
$c(x)=\sum_{i=0}^{n-1} c_i x^i$, where all 
$c_i \in \gf(q)$. 
A primitive cyclic code $\C$ is called \emph{even-like} if $1$ is a zero of its generator polynomial, 
and \emph{odd-like} otherwise. 

Let $\cJ$ and $\cJ^*$ denote $\gf(q^m)$ and $\gf(q^m)^*$, respectively. Let $\alpha$ be a primitive element 
of $\gf(q^m)$. The set $\cJ$ will be the index set of the extended cyclic codes of length $q^m$, and the set 
$\cJ^*$ will be the index set of the cyclic codes of length $n$. Let $X$ be an indeterminate. Define 
\begin{eqnarray}
\gf(q)[\cJ]=\left\{a=\sum_{g \in \cJ} a_g X^g: a_g \in \gf(q) \mbox{ for all } g \in \cJ \right\}. 
\end{eqnarray}    
The set $\gf(q)[\cJ]$ is an algebra under the following operations 
\begin{eqnarray*}
u \sum_{g \in \cJ} a_g X^g + v \sum_{g \in \cJ} b_g X^g = \sum_{g \in \cJ} (ua_g +v b_g) X^g  
\end{eqnarray*} 
for all $u, \, v \in \gf(q)$, and 
\begin{eqnarray}
\left(\sum_{g \in \cJ} a_g X^g \right) \left(\sum_{g \in \cJ} b_g X^g \right) 
= \sum_{g \in \cJ} \left(\sum_{h \in \cJ} a_h b_{g-h} \right) X^g. 
\end{eqnarray} 
The zero and unit of $\gf(q)[\cJ]$ are $\sum_{g \in \cJ} 0 X^g$ and $X^0$, respectively. 

Similarly, let 
\begin{eqnarray}
\gf(q)[\cJ^*]=\left\{a=\sum_{g \in \cJ^*} a_g X^g: a_g \in \gf(q) \mbox{ for all } g \in \cJ^* \right\}. 
\end{eqnarray}   
The set $\gf(q)[\cJ^*]$ is not a subalgebra, but a subspace of $\gf(q)[\cJ]$. Obviously, the elements of 
$\gf(q)[\cJ^*]$ are of the form 
$$ 
\sum_{i=0}^{n-1} a_{\alpha^i} X^{\alpha^i}, 
$$  
and those of 
$\gf(q)[\cJ]$ are of the form 
$$ 
a_0X^0 + \sum_{i=0}^{n-1} a_{\alpha^i} X^{\alpha^i}.  
$$  
Subsets of the subspace $\gf(q)[\cJ^*]$ will be used to characterise primitive cyclic codes over $\gf(q)$ and 
those of the algebra $\gf(q)[\cJ]$ will be employed to characterise extended primitive cyclic codes over $\gf(q)$. 

We define a one-to-one correspondence between $\cR_n$ and $\gf(q)[\cJ^*]$ by 
\begin{eqnarray}\label{eqn-Upsilon}
\Upsilon: c(x)=\sum_{i=0}^{n-1} c_i x^i \to C(X)=\sum_{i=0}^{n-1} C_{\alpha^i} X^{\alpha^i},
\end{eqnarray} 
where $C_{\alpha^i}=c_i$ for all $i$. 

The following theorem is obviously true. 

\begin{theorem}\label{thm-newCharacterisationCyclicCodes}
$\C \subseteq \cR_n$ has the circulant cyclic shift property if and only if $\Upsilon(\C) \subseteq \gf(q)[\cJ^*]$ has the property 
that 
$$ 
\sum_{i=0}^{n-1} C_{\alpha^i} X^{\alpha^i} = \sum_{g \in \cJ^*} C_g X^g \in \Upsilon(\C)  
$$ 
if and only if 
$$ 
\sum_{i=0}^{n-1} C_{\alpha^i} X^{\alpha \alpha^i} = \sum_{g \in \cJ^*} C_g X^{\alpha g} \in \Upsilon(\C)  
$$ 
\end{theorem} 

With Theorem \ref{thm-newCharacterisationCyclicCodes}, every primitive cyclic code over $\gf(q)$ can be  
viewed as a special 
subset of $\gf(q)[\cJ^*]$ having the property documented in this theorem. This special representation of 
primitive cyclic codes over $\gf(q)$ will be very useful for determining a subgroup of the automorphism group 
of certain primitive cyclic codes.  

It is now time to extend primitive cyclic codes, which are subsets of $\gf(q)[\cJ^*]$. We use the element $0 
\in \cJ$ to index the extended coordinate. 
The extended codeword 
$\overline{C}(X)$ of a codeword $C(X)=\sum_{g \in \cJ^*} C_g X^g$ in $\gf(q)[\cJ^*]$ is defined by 
\begin{eqnarray}
\overline{C}(X)=\sum_{g \in \cJ} C_g X^g
\end{eqnarray} 
with $\sum_{g \in \cJ} C_g=0.$

Notice that $X^{\alpha 0}=X^0=1$. The following then follows from Theorem \ref{thm-newCharacterisationCyclicCodes}. 

\begin{theorem}\label{thm-newCharacterisationExtendCyclicCodes}
The extended code $\overline{\C}$ of a cyclic code $\C \subseteq \gf(q)[\cJ^*]$ is a subspace of 
$\gf(q)[\cJ]$ such that 
\begin{eqnarray*}
 \overline{C}(X)=\sum_{g \in \cJ} C_g X^g \in \overline{\C} \mbox{ if and only if } 
 \sum_{g \in \cJ} C_g X^{\alpha g} \in \overline{\C} \mbox{ and } 
\sum_{g \in \cJ} C_g=0. 
\end{eqnarray*} 
\end{theorem}

If a cyclic code $\C$ is viewed as an ideal of $\cR_n=\gf(q)[x]/(x^n-1)$, it can be defined by its set of zeros 
or its defining set. When $\C$ and $\overline{\C}$ are put in the 
settings $\gf(q)[\cJ^*]$ and $\gf(q)[\cJ]$, respectively, they can be defined with some counterpart of the defining set. 
This can be done with the assistance of the following function $\phi_s$ from $\gf(q)[\cJ]$ to $\cJ$: 
\begin{eqnarray}
\phi_s\left(\sum_{g \in \cJ} C_g X^g\right)= \sum_{g \in \cJ} C_g g^s, 
\end{eqnarray}    
where $s \in \overline{\cN}:=\{i: 0 \leq i \leq n\}$ and by convention $0^0=1$ in $\cJ$. 

The following follows from Theorem \ref{thm-newCharacterisationExtendCyclicCodes} and the definition of 
$\phi_s$ directly. 

\begin{lemma}\label{lem-dec101}  
$\overline{C}(X)$ is the extended codeword of $C(X) \in \gf(q)[\cJ^*]$ if and only if $\phi_0(\overline{C}(X))=0$. 
In particular, if $\overline{\C}$ is the extended code of a primitive cyclic code $\C \subseteq \gf(q)[\cJ^*]$, 
then $\phi_0(\overline{C}(X))=0$ for all $\overline{C}(X) \in \overline{\C}$. 
\end{lemma}

\begin{lemma}\label{lem-dec102} 
Let $\C$ be a primitive cyclic code of length $n$ over $\gf(q)$. Let $T$ be the defining set of $\C$ with 
respect to $\alpha$, when it is viewed as an ideal of $\cR_n$. Let $s \in T$ and $1 \leq s \leq n-1$.  
We have 
then $\phi_s(\overline{C}(X))=0$ for all $\overline{C}(X) \in \overline{\C}$. 
\end{lemma}

\begin{lemma}\label{lem-dec103}  
Let $\C$ be a primitive cyclic code of length $n$ over $\gf(q)$. Let $T$ be the defining set of $\C$ with 
respect to $\alpha$, when it is viewed as an ideal of $\cR_n$. Then $0 \in T$ if and only if   
$\phi_n(\overline{C}(X))=0$ for all $\overline{C}(X) \in \overline{\C}$. 
\end{lemma}

Combining Lemmas \ref{lem-dec101}, \ref{lem-dec102}, \ref{lem-dec103} and the discussions above, 
we can define an extended cyclic code in terms of a defining set as follows. 

A code $\overline{\C}$ of length $q^m$ is an \emph{extended primitive cyclic code}\index{extended primitive cyclic code} 
with definition set $\overline{T}$ provided $\overline{T} \setminus \{n\} \subseteq \overline{\cN}$ is a union of 
$q$-cyclotomic cosets modulo $n=q^m-1$ with $0 \in \overline{T}$ and 
\begin{eqnarray}\label{eqn-defdefCodes}
\overline{\C}= \left\{\overline{C}(X) \in \gf(q)[\cJ]: \phi_s(\overline{C}(X))=0 \mbox{ for all } s \in \overline{T} \right\}. 
\end{eqnarray} 

The following remarks are helpful for fully understanding the characterisation of extended primitive 
cyclic codes: 
\begin{itemize} 
\item The condition that $\overline{T} \setminus \{n\} \subseteq \overline{\cN}$ is a union of 
$q$-cyclotomic cosets modulo $n=q^m-1$ is to ensure that the code $\C$ obtained by puncturing the 
first coordinate of $\overline{\C}$ and ordering the elements of $\cJ$ with $(0, \alpha^n, \alpha^1, \cdots, \alpha^{n-1})$ is a primitive cyclic code. 
\item The additional requirement $0 \in \overline{T}$ and (\ref{eqn-defdefCodes}) are to make sure 
that $\overline{\C}$ is the extended code of $\C$. 
\item If $n \in \overline{T}$, then $\C$ is an even-like code. In this case, the extension is trivial, 
i.e., the extended coordinate in every codeword of $\overline{\C}$ is always equal to $0$. 
If $n \not\in \overline{T}$, then $0 \not\in T$. Thus, the extension is nontrivial. 
\item If $\overline{\C}$ is the extended code of a primitive cyclic code $\C$, then 
\begin{eqnarray*}
\overline{T}=\left\{ 
\begin{array}{ll}
\{0\} \cup T & \mbox{ if } 0 \not\in T, \\
\{0, n\} \cup T & \mbox{ if } 0 \in T.  
\end{array}
\right.
\end{eqnarray*}
where $T$ and $\overline{T}$ are the defining sets of $\C$ and $\overline{\C}$, respectively.  
\item The following diagram illustrates the relations among the two codes and their definition sets: 
\begin{eqnarray*}
\begin{array}{lcr}
\C \subseteq \cR_n & \Longleftrightarrow \C \subseteq \gf(q)[\cJ^*] 
                      \Longrightarrow & \gf(q)[\cJ] \supseteq \overline{\C} \\
T \subseteq \cN     &      &    \overline{T} \subseteq \overline{\cN}                       
\end{array} 
\end{eqnarray*}
\end{itemize}   

Let $\sigma$ be a permutation on $\cJ$. This permutation acts on a code $\overline{\C} \subseteq \gf(q)[\cJ]$ as follows: 
\begin{eqnarray}
\sigma\left( \sum_{g \in \cJ} C_g X^g \right) = \sum_{g \in \cJ} C_g X^{\sigma(g)}.   
\end{eqnarray}

The \emph{affine permutation group}\index{affine permutation group}, denoted by $\AGL(1, q^m)$, is 
defined by 
\begin{eqnarray}\label{eqn-AGL(1qm)}
\AGL(1, q^m)=\{\sigma_{(a,b)}(y)=ay+b: a \in \cJ^*, \, b \in \cJ\}.  
\end{eqnarray} 
We have the following conclusions about $\AGL(1, q^m)$ whose proofs are straightforward: 
\begin{itemize} 
\item $\AGL(1, q^m)$ is a permutation group on $\cJ$ under the function composition. 
\item The group action of $\AGL(1, q^m)$ on $\gf(q^m)$ is doubly transitive, i.e., $2$-transitive. 
\item $\AGL(1, q^m)$ has order $(n+1)n=q^m(q^m-1)$. 
\item Obviously, the maps $\sigma_{(a,0)}$ are merely the cyclic shifts on the coordinates 
$(\alpha^n, \alpha^1, \cdots, \alpha^{n-1})$ each fixing the coordinate $0$. 
\end{itemize} 

An \emph{affine-invariant code}\index{affine-invariant} is an extended primitive cyclic code 
$\overline{\C}$ such that $\AGL(1, q^m) \subseteq \PAut(\overline{\C})$. For certain applications, 
it is important to know if a given extended primitive cyclic code $\overline{\C}$ is affine-invariant or not. 
This question can be answered by examining the defining set of the code. In order to do this, we introduce a partial 
ordering $\preceq$ on $\overline{\cN}$. Suppose that $q=p^t$ for some positive integer $t$. Then by 
definition $\overline{\cN}=\{0,1,2, \cdots, n\}$, where $n=q^m-1=p^{mt}-1$. The $p$-adic expansion of 
each $s \in \overline{\cN}$ is given by 
\begin{eqnarray*}
s=\sum_{i=0}^{mt-1} s_i p^i, \, \mbox{ where } 0 \leq s_i <p \mbox{ for all } 0 \leq i \leq mt-1.   
\end{eqnarray*}    
Let the $p$-adic expansion of $r \in \overline{\cN}$ be 
$$ 
r=\sum_{i=0}^{mt-1} r_i p^i. 
$$  
We say that $r \preceq s$ if $r_i \leq s_i$ for all $0 \leq i \leq mt-1$. By definition, we have 
$r\leq s$ if $r \preceq s$.

The following is a characterisation of affine-invariant codes due to Kasami, Lin and Peterson \cite{KLP68}. 

\begin{theorem}[Kasami-Lin-Peterson]\label{thm-KLPtheorem}
Let $\overline{\C}$ be an extended cyclic code of length $q^m$ over $\gf(q)$ with defining set $\overline{T}$. 
The code $\overline{\C}$ is affine-invariant if and only if whenever $s \in \overline{T}$ then $r \in \overline{T}$ for all 
$r \in \overline{\cN}$ with $r \preceq s$.    
\end{theorem}

Theorem \ref{thm-KLPtheorem} will be employed in the next section. It is a very useful tool to prove that 
an extended primitive cyclic code is affine-invariant.  

It is straightforward to prove that $\AGL(1, q^m)$ is doubly transitive on $\gf(q^m)$. 
The following theorem then follows from Theorem \ref{thm-designCodeAutm}.  

\begin{theorem}\label{thm-mainthmjan18}
Let $\C$ be an extended cyclic code of length $q^m$ over $\gf(q)$. If $\C$ affine-invariant, 
then the supports of the codewords of weight $k$ in $\C$ form a $2$-design, provided that 
$A_k \neq 0$.  
\end{theorem} 

The following is a list of known affine-invariant codes. 
\begin{itemize}
\item The classical generalised Reed-Muller codes of length $q^n$ \cite{AK92}. 
\item A family of newly generalised Reed-Muller codes of length $q^n$ \cite{DLX16}.   
\item The narrow-sense primitive BCH codes.  
\end{itemize}

If new affine-invariant codes are discovered, new $2$-designs may be obtained. In the next section, 
we will present a type of affine-invariant binary codes of length $2^m$, and will investigate their 
designs. Our major objective is to construct an infinite family of Steiner systems $S(2,4, 2^m)$.   

\section{A type of affine-invariant codes and their designs} 

In this section, we first present a class of affine-invariant binary codes of length $2^m$, 
and then study their designs. Our main purpose is to present an infinite family of Steiner 
systems $S(2,4,2^m)$ for every $m \equiv 2 \pmod{4} \geq 6$.  

Let $b$ denote the number of blocks in a $t$-$(v, k, \lambda)$ design. It is easily seen that  
\begin{eqnarray}\label{eqn-tdesignblocknum}
b= \lambda \frac{\binom{v}{t}}{\binom{k}{t} }. 
\end{eqnarray} 

We will need the following lemma 
in subsequent sections, which is a variant of the MacWilliam Identity \cite[p. 41]{vanLint}. 

\begin{theorem} \label{thm-MI}
Let $\C$ be a $[v, \kappa, d]$ code over $\gf(q)$ with weight enumerator $A(z)=\sum_{i=0}^v A_iz^i$ and let
$A^\perp(z)$ be the weight enumerator of $\C^\perp$. Then
$$A^\perp(z)=q^{-\kappa}\Big(1+(q-1)z\Big)^vA\Big(\frac {1-z} {1+(q-1)z}\Big).$$
\end{theorem} 

Shortly, we will need also the following theorem. 

\begin{theorem}\label{thm-allcodes}
Let $\C$ be an $[n, k, d]$ binary linear code, and let $\C^\perp$ denote the dual of $\C$. Denote by $\overline{\C^\perp}$ 
the extended code of $\C^\perp$, and let  $\overline{\C^\perp}^\perp$ denote the dual of  $\overline{\C^\perp}$. Then we 
have the following. 
\begin{enumerate}
\item $\C^\perp$ has parameters $[n, n-k, d^\perp]$, where $d^\perp$ denotes the minimum distance of $\C^\perp$. 
\item $\overline{\C^\perp}$ has parameters $[n+1, n-k, \overline{d^\perp}]$, where $\overline{d^\perp}$ denotes the minimum 
distance of $\overline{\C^\perp}$, and is given by 
\begin{eqnarray*}
\overline{d^\perp} = \left\{ 
\begin{array}{ll}
d^\perp  & \mbox{ if $d^\perp$ is even,}  \\
d^\perp  + 1 & \mbox{ if $d^\perp$ is odd.} 
\end{array} 
\right. 
\end{eqnarray*}
\item $\overline{\C^\perp}^\perp$ has parameters $[n+1, k+1, \overline{d^\perp}^\perp]$, where $\overline{d^\perp}^\perp$ denotes the minimum 
distance of $\overline{\C^\perp}^\perp$. 
Furthermore, $\overline{\C^\perp}^\perp$ has only even-weight codewords, and all the nonzero weights in $\overline{\C^\perp}^\perp$ are 
the following: 
\begin{eqnarray*}
w_1, \, w_2,\, \cdots,\, w_t;\, n+1-w_1, \, n+1-w_2, \, \cdots, \, n+1-w_t; \, n+1,  
\end{eqnarray*} 
where $w_1,\,  w_2,\, \cdots,\, w_t$ denote all the nonzero weights of $\C$. 
\end{enumerate} 
\end{theorem} 

\begin{proof}
The conclusions of the first two parts are straightforward. We prove only the conclusions of the third part below. 

Since $\overline{\C^\perp}$ has length $n+1$ and dimension $n-k$,  the dimension of $\overline{\C^\perp}^\perp$ 
is $k+1$. By assumption, all codes under consideration are binary. By definition, $\overline{\C^\perp}$ has only 
even-weight codewords. Recall that $\overline{\C^\perp}$ is the extended code of $\C^\perp$. It is known that  
the generator matrix of $\overline{\C^\perp}^\perp$ is given by  (\cite[p. 15]{HP03})  
\begin{eqnarray*}
\left[ 
\begin{array}{cc}
\bone & 1 \\
G & \bzero
\end{array}
\right]. 
\end{eqnarray*}
where $\bone=(1 1 1 \cdots 1)$ is the all-one vector of length $n$,  $\bzero=(000 \cdots 0)^T$, which is a column 
vector of length $n$, and $G$ is the generator matrix of $\C$. Notice again that  $\overline{\C^\perp}^\perp$ is binary, the 
desired conclusions on the weights in $\overline{\C^\perp}^\perp$ follow from the relation between the two generator 
matrices of the two codes  $\overline{\C^\perp}^\perp$ and $\C$. 
\end{proof}

\subsection{The type of affine-invariant codes and their designs}

Starting from now on, we deal with only binary codes and their support designs, 
and we define $n=2^m-1$ and $\bar{n}=2^m$. 

Let $m \geq 2$ be a positive integer. Define $\overline{m}=\lfloor m/2 \rfloor$ and 
$M=\{1,2, \cdots, \overline{m}\}$. Let $E$ be any nonempty subset of $M$. Let 
\begin{equation}\label{eqn-generatorpolyCe}
g_E(x)=\m_{\alpha}(x) \lcm\{\m_{\alpha^{1+2^e}}(x): e \in E\}, 
\end{equation} 
where $\alpha$ is a generator of $\gf(2^m)^*$, $\m_{\alpha^i}(x)$ denotes the minimal polynomial 
of $\alpha^i$ over $\gf(2)$, and $\lcm$ denotes the least common multiple of a set of polynomials. 
Note that every $e \in E$ satisfies $e \leq \overline{m}$, and the $2$-cyclotomic cosets $C_1$ and $C_e$ 
are disjoint. Consequently, the two irreducible polynomials $\m_{\alpha}(x)$ and $\m_{\alpha^{1+2^e}}(x)$ 
are relatively prime. It then follows that $g_E(x)$ divides $x^n-1$. Let $\C_E$ denote the binary 
cyclic code of length $n$ with generator polynomial $g_{E}(x)$. 

\begin{theorem}
Let $m \geq 3$. Then the generator polynomial of $\C_E$ is given by  
$$ 
g_E(x)=\m_{\alpha}(x) \prod_{e \in E} \m_{\alpha^{1+2^e}}(x). 
$$ 
Furthermore, $\C_E$ has dimension  
\begin{eqnarray}
\dim(\C_E)=
\left\{ 
\begin{array}{ll} 
2^m-1- (2|E|+1)m/2 & \mbox{ if $m$ is even and $m/2 \in E$.} \\ 
2^m-1-(|E|+1)m  & \mbox{ otherwise,} 
\end{array}
\right. 
\end{eqnarray}
\end{theorem}   

\begin{proof}
The following list of properties was proved in \cite{YH96}: 
\begin{itemize}
\item For each $e \in E$, $1+2^e$ is a coset leader. 
\item For each $e \in E$, $|C_{e}|=m$, except that $m$ is even and $e=m/2$, in which case 
      $|C_{m/2}|=m/2$.  
\end{itemize} 
Note that $1$ is the coset leader of the $2$-cyclotomic coset $C_1$ with $|C_1|=m$. Then the  
desired conclusions on the generator polynomial and dimension follow. 
\end{proof}

\begin{theorem}\label{thm-jan182}
The extended code $\overline{C_E}$ is affine invariant.  
\end{theorem} 

\begin{proof}
We prove the desired conclusion with the help Theorem \ref{thm-KLPtheorem} and follow the notation and 
symbols employed in the proof of Theorem \ref{thm-KLPtheorem}. 
Let $\overline{\cN}=\{0,1,2, \cdots, n\}$, where $n=2^m-1$. 
The defining set $T$ of the 
cyclic code $\C_E$ is $T=C_1 \cup (\cup_{e \in E} C_e)$. Since $0 \not\in T$, the defining 
set $\overline{T}$ of $\overline{C_E}$ is given by 
$$ 
\overline{T}=C_1 \cup (\cup_{e \in E} C_e) \cup \{0\}. 
$$
Let $s \in \overline{T}$ and $r \in \overline{\cN}$. Assume 
that $r \preceq s$. We need prove that $r \in \overline{T}$ by Theorem \ref{thm-KLPtheorem}. 

If $r=0$, then obviously $r \in \overline{T}$. Consider now the case $r >0$. In this case $s \geq r \geq 1$. 
If $s \in C_1$, then the Hamming weight $\wt(s)=1.$ As $r \preceq s$, $r=s$. Consequently, $r \in 
C_1 \subset \overline{T}$. 
If $s \in C_e$, then the Hamming weight $\wt(s)=2.$ As $r \preceq s$, either $\wt(r)=1$ or $r=s$. 
In bother cases,  $r \in \overline{T}$. The desired conclusion then follows from Theorem \ref{thm-KLPtheorem}. 
\end{proof} 

Combining Theorems \ref{thm-jan182}, \ref{thm-mainthmjan18} and \ref{thm-dualdesignbinary}, we 
arrive at the following conclusions. 

\begin{theorem}\label{thm-mymainjan18}
Let $m \geq 3$ be an integer. The supports of the codewords of every weight $k$ in $\overline{\C_{E}}$ 
(respectively, $\overline{\C_{E}}^\perp$) form a $2$-design, provided that $\overline{A}_k \neq 0$ 
(respectively, $\overline{A}_k^\perp \neq 0$).    
\end{theorem} 

Theorem \ref{thm-mymainjan18} includes a class of $2^{\lfloor m/2 \rfloor}-1$ affine invariant 
binary codes $\overline{\C_E}$ and their duals. They give exponentially many infinite families 
of $2$-$(2^m, k, \lambda)$ designs. To determine the parameters $(2^m, k, \lambda)$ of the $2$-designs, 
we need to settle the weight distributions of these codes. The weight distributions of these codes 
are related to quadratic form, bilinear forms, and alternating bilinear forms, and are open in 
general. Note that the code $\C_{E}$ may be a BCH code in some cases, but is not a BCH code in most cases. 

\subsection{Designs from the codes $\C_{\{1+2^e\}}$ and their relatives}\label{sec-mainsect} 

As made clear earlier, our main objective is to construct an infinite family of Steiner systems 
$S(2, 4, 2^m)$. To this end, we consider the code $\C_E$ and its extended code $\overline{\C_E}$ in this 
section for the special case $E=\{1+2^e\}$, where $1 \leq e \leq \overline{m}=\lfloor m/2 \rfloor$. 
For simplicity, we denote this code by $\C_e$ in this section.   

\begin{table}[ht]
\center 
\caption{Weight distribution I}\label{tab-aaCG1} 
{
\begin{tabular}{ll}
\hline
Weight $w$    & No. of codewords $A_w$  \\ \hline
$0$                                                        & $1$ \\ 
$2^{m-1}-2^{m-1-h}$           & $(2^m-1)(2^h+1)2^{h-1}$ \\ 
$2^{m-1}$                             & $(2^m-1)(2^m-2^{2h}+1)$ \\ 
$2^{m-1}+2^{m-1-h}$           & $(2^m-1)(2^h-1)2^{h-1}$ \\ \hline
\end{tabular}
}
\end{table}

\begin{table}[ht]
\center 
\caption{Weight distribution II}\label{tab-aaCG2}
{
\begin{tabular}{ll}
\hline
Weight $w$    & No. of codewords $A_w$  \\ \hline
$0$                                                        & $1$ \\ 
$2^{m-1}-2^{(m-2)/2}$          & $(2^{m/2}-1)(2^{m-1}+2^{(m-2)/2})$ \\ 
$2^{m-1}$                               & $2^m-1$ \\ 
$2^{m-1}+2^{(m-2)/2}$          & $(2^{m/2}-1)(2^{m-1}-2^{(m-2)/2})$ \\ \hline
\end{tabular}
}
\end{table}

\begin{table}[ht]
\center 
\caption{Weight distribution III}\label{tab-aaMine}
{
\begin{tabular}{ll}
\hline
Weight $w$    & No. of codewords $A_w$  \\ \hline
$0$                                                        & $1$ \\ 
$2^{m-1}-2^{(m+\ell-2)/2}$          & $2^{(m-\ell-2)/2}(2^{(m-\ell)/2}+1)(2^m-1)/(2^{\ell/2}+1)$ \\ 
$2^{m-1}-2^{(m-2)/2}$          & $2^{(m+\ell-2)/2}(2^{m/2}+1)(2^m-1)/(2^{\ell/2}+1)$ \\ 
$2^{m-1}$                               & $( (2^{\ell/2} -1)2^{m-\ell} +1 )(2^m - 1)$ \\ 
$2^{m-1}+2^{(m-2)/2}$          & $2^{(m+\ell-2)/2}(2^{m/2}-1)(2^m-1)/(2^{\ell/2}+1)$ \\ 
$2^{m-1}+2^{(m+\ell-2)/2}$          & $2^{(m-\ell-2)/2}(2^{(m-\ell)/2}-1)(2^m-1)/(2^{\ell/2}+1)$ \\ \hline 
\end{tabular}
}
\end{table}

The following theorem provides information on the parameters of $\C_e$ and its dual 
$\C_{e}^\perp$ \cite{Kasa69}. 

\begin{theorem}\label{thm-calder} 
Let $m \ge 4$ and $1 \leq e \le m/2$. Then $\C_{e}^\perp$ is a three-weight
code if and only if either $m/\gcd(m,e)$ is odd or $m$ is even and $e=m/2$, where
$n=2^m-1$.

When $m/\gcd(m,e)$ is odd, define $h=(m-\gcd(m,e))/2$. Then the dimension of
$\C_{e}^\perp$ is $2m$, and the weight
distribution of $\C_{e}^\perp$ is given in Table \ref{tab-aaCG1}. The code $\C_{e}$ has
parameters $[n, n-2m, d]$, where
$$
d = \left\{ \begin{array}{ll}
                             3 & \mbox{ if } \gcd(e,m)>1; \\
                             5 & \mbox{ if } \gcd(e,m)=1.
\end{array}
\right.
$$

When $m$ is even and $e=m/2$, the dimension of  $\C_{e}^\perp$ is $3m/2$ and the weight
distribution of $\C_{e}^\perp$ is given in Table \ref{tab-aaCG2}. The code $\C_{e}$ has
parameters $[n, n-3m/2, 3]$. 

When $m/\gcd(m,e)$ is even and $1 \leq e < m/2$, $\C_{e}^\perp$ has dimension $2m$ and the weight 
distribution in Table \ref{tab-aaMine}, where $\ell=2\gcd(m, e)$, and $\C_{e}$ has parameters $[n, n-2m, d]$, where 
$$
d = \left\{ \begin{array}{ll}
                             3 & \mbox{ if } \gcd(e,m)>1; \\
                             5 & \mbox{ if } \gcd(e,m)=1.
\end{array}
\right.
$$ 
\end{theorem} 

The weight distributions of the code $\C_{e}^\perp$ documented in Theorem \ref{thm-calder} were 
indeed proved by Kasami in \cite{Kasa69}. However, the conclusions on the minimum distance $d$ 
of $\C_{e}$ were stated in \cite{Kasa69} without being proved. We inform the reader that they can 
be proved with the proved weight distribution of $\C_{e}^\perp$ and Theorem \ref{thm-MI}, though 
the details of proof are tedious in some cases.  

We would find the parameters of the $2$-designs held in the codes $\overline{C_e}$ and 
$\overline{C_e}^\perp$, and need to know the weight distributions of these two codes, 
which can be derived from those of the code $\C_{e}^\perp$ described in Theorem \ref{thm-calder}. 
We first determine the weight distribution of $\overline{C_e}^\perp$.

\begin{table}[ht]
\center 
\caption{Weight distribution IV}\label{tab-aaCG1jan18} 
{
\begin{tabular}{ll}
\hline
Weight $w$    & No. of codewords $A_w$  \\ \hline
$0$                                                        & $1$ \\ 
$2^{m-1}-2^{m-1-h}$           & $(2^m-1)2^{2h}$ \\ 
$2^{m-1}$                             & $(2^m-1)(2^{m+1}-2^{2h+1}+2)$ \\ 
$2^{m-1}+2^{m-1-h}$           & $(2^m-1)2^{2h}$ \\  
$2^m$                           &  $1$ \\ \hline
\end{tabular}
}
\end{table}

\begin{table}[ht]
\center 
\caption{Weight distribution V}\label{tab-aaCG2jan18}
{
\begin{tabular}{ll}
\hline
Weight $w$    & No. of codewords $A_w$  \\ \hline
$0$                                                        & $1$ \\ 
$2^{m-1}-2^{(m-2)/2}$          & $(2^{m/2}-1)2^{m}$ \\ 
$2^{m-1}$                               & $2^{m+1}-2$ \\ 
$2^{m-1}+2^{(m-2)/2}$          & $(2^{m/2}-1)2^{m}$ \\ 
$2^m$                          & $1$ \\ \hline
\end{tabular}
}
\end{table}

\begin{table}[ht]
\center 
\caption{Weight distribution VI}\label{tab-aaMinejan18}
{
\begin{tabular}{ll}
\hline
Weight $w$    & No. of codewords $A_w$  \\ \hline
$0$                                                        & $1$ \\ 
$2^{m-1}-2^{(m+\ell-2)/2}$          & $2^{m-\ell} (2^m-1)/(2^{\ell/2}+1)$ \\ 
$2^{m-1}-2^{(m-2)/2}$          & $2^{(2m+\ell)/2} (2^m-1)/(2^{\ell/2}+1)$ \\ 
$2^{m-1}$                               & $2( (2^{\ell/2} -1)2^{m-\ell} +1 )(2^m - 1)$ \\ 
$2^{m-1}+2^{(m-2)/2}$          & $2^{(2m+\ell)/2} (2^m-1)/(2^{\ell/2}+1)$ \\ 
$2^{m-1}+2^{(m+\ell-2)/2}$          & $2^{m-\ell} (2^m-1)/(2^{\ell/2}+1)$ \\ 
$2^m$                               & $1$ \\ \hline 
\end{tabular}
}
\end{table}

The following theorem provides information on the parameters of $\overline{\C_e}$ and its dual 
$\overline{\C_{e}}^\perp$. 

\begin{theorem}\label{thm-calderjan18} 
Let $m \ge 4$ and $1 \leq e \le m/2$. 
When $m/\gcd(m,e)$ is odd, define $h=(m-\gcd(m,e))/2$. Then $\overline{\C_{e}}^\perp$ 
has parameters $[2^m, 2m+1, 2^{m-1}-2^{m-1-h}]$, and the weight
distribution in Table \ref{tab-aaCG1jan18}. 
The parameters of $\overline{\C_{e}}$ are $[2^m, 2^m-1-2m, \overline{d}]$, 
where 
$$
\overline{d} = \left\{ \begin{array}{ll}
                             4 & \mbox{ if } \gcd(e,m)>1; \\
                             6 & \mbox{ if } \gcd(e,m)=1.
\end{array}
\right.
$$

When $m$ is even and $e=m/2$, $\overline{\C_{e}}^\perp$ has parameters $[2^m, 1+3m/2, 2^{m-1}-2^{(m-2)/2}]$ and the weight
distribution in Table \ref{tab-aaCG2jan18}. 
The code $\overline{\C_{e}}$ has parameters $[2^m, 2^m-1-3m/2, 4]$. 

When $m/\gcd(m,e)$ is even and $1 \leq e < m/2$, $\overline{\C_{e}}^\perp$ has parameters
$$[2^m, \, 2m+1,\, 2^{m-1}-2^{(m+\ell-2)/2}]$$ and the weight 
distribution in Table \ref{tab-aaMinejan18}, where $\ell=2\gcd(m, e)$, and $\overline{\C_{e}}$ has parameters $[2^m, 2^m-1-2m, \overline{d}]$, where 
$$
\overline{d} = \left\{ \begin{array}{ll}
                             4 & \mbox{ if } \gcd(e,m)>1; \\
                             6 & \mbox{ if } \gcd(e,m)=1.
\end{array}
\right.
$$ 
\end{theorem} 

\begin{proof}
We prove only the conclusions of the first part. The conclusions of the other parts can be proved 
similarly. 

Consider now the case that $m/\gcd(m,e)$ is odd. Since the minimum weight of $\C_e$ is odd,  
the minimum distance of $\overline{\C_e}$ is one more than that of $\C_e$. This proves the 
conclusion on the minimum distance of $\overline{\C_e}$. By definition, $\dim(\C_e)=\dim(\overline{\C_e})$, and the length of $\overline{\C_e}$ is $\bar{n}=n+1=2^m$. 

The dimension of $\overline{\C_e}^\perp$ follows from that of $\overline{\C_e}$. It remains 
to prove the weight distribution of $\overline{\C_e}^\perp$. By definition, $\overline{\C_e}$ 
has only even weights. It then follows that the all-one vector is a codeword of 
$\overline{\C_e}^\perp$. Then by Theorems \ref{thm-allcodes} 
and \ref{thm-calder}, $\overline{\C_e}^\perp$ has all the following weights 
$$ 
2^{m-1} \pm 2^{m-1-h}, \ 2^{m-1} \pm 2^{(m-2)/2}, \ 2^{m-1}, \ 2^m. 
$$ 
Due to symmetry of weights and the existence of the all-one vector in $\overline{\C_e}^\perp$, 
$$ 
A_{2^{m-1} + 2^{m-1-h}}=A_{2^{m-1} - 2^{m-1-h}}, \ 
A_{2^{m-1} + 2^{(m-2)/2}}=A_{2^{m-1} - 2^{(m-2)/2}}. 
$$ 
Note that the minimum distance of $\overline{\C_e}$ is $4$ or $6$. Solving the first four 
Pless power moments yields the frequencies of all the weights. 
 
\end{proof}

Combining Theorem \ref{thm-mymainjan18} and (\ref{eqn-tdesignblocknum}), we deduce the 
following. 
 
\begin{theorem}\label{thm-designsdualjan18} 
Let $m \ge 4$ and $1 \leq e \le m/2$. 
When $m/\gcd(m,e)$ is odd, define $h=(m-\gcd(m,e))/2$. Then $\overline{\C_{e}}^\perp$ 
holds a $2$-$(2^m, k, \lambda)$ design for the following pairs $(k, \lambda)$: 
\begin{itemize}
\item $(k, \lambda)=\left(2^{m-1} \pm 2^{m-1-h},\  (2^{2h-1}\pm 2^{h-1})(2^{m-1} \pm 2^{m-1-h}-1)\right)$.  
\item $(k, \lambda)=\left(2^{m-1}, \  (2^{m-1}-1)(2^{m}-2^{2h}+1) \right).$ 
\end{itemize}

When $m$ is even and $e=m/2$, $\overline{\C_{e}}^\perp$ 
holds a $2$-$(2^m, k, \lambda)$ design for the following pairs $(k, \lambda)$: 
\begin{itemize}
\item $(k, \lambda)=\left(2^{m-1} \pm 2^{(m-2)/2}, \ 
       2^{(m-2)/2}(2^{m/2} - 1)(2^{(m-2)/2} \pm 1) \right).$ 
\item $(k, \lambda)=\left(2^{m-1}, \ 2^{m-1}-1\right).$ 
\end{itemize}

When $m/\gcd(m,e)$ is even and $1 \leq e < m/2$, $\overline{\C_{e}}^\perp$ 
holds a $2$-$(2^m, k, \lambda)$ design for the following pairs $(k, \lambda)$: 
\begin{itemize}
\item $(k, \lambda)=\left(2^{m-1} \pm 2^{(m+\ell-2)/2}, \ 
            \frac{(2^{m-1} \pm 2^{(m+\ell-2)/2})(2^{m-1} \pm 2^{(m+\ell-2)/2}-1)}{2^\ell(2^{\ell/2}+1)} \right),$ 
\item $(k, \lambda)=\left(2^{m-1} \pm 2^{(m-2)/2}, \ 
       \frac{2^{(m+\ell-2)/2}(2^{m/2} \pm 1)(2^{m-1} \pm 2^{(m-2)/2}-1)}{2^{\ell/1}-1} \right),$ 
\item $(k, \lambda)=\left(2^{m-1}, \ ( (2^{\ell/2} -1)2^{m-\ell} +1 )(2^{m-1} - 1) \right),$ 
\end{itemize}  
where $\ell=2\gcd(m, e)$. 
\end{theorem} 

To determine the parameters of the $2$-designs held in the extended code $\overline{\C_{e}}$, we need to 
find out the weight distribution of $\overline{\C_{e}}$. In theory, the weight distribution of 
$\overline{\C_{e}}$ can be settled using the weight enumerator of $\overline{\C_{e}}^\perp$ 
given in Tables \ref{tab-aaCG1jan18}, \ref{tab-aaCG2jan18}, and \ref{tab-aaMinejan18}. However, 
it is practically hard to find a simple expression of the weight distribution of $\overline{\C_{e}}$. 

In the rest of this section, we consider only the weight distribution of $\overline{\C_{e}}$ 
in a special case, in order to construct an infinite family of Steiner systems $S(2, 4, 2^m)$ 
for all $m \equiv 2 \pmod{4}$. 

As a special case of Theorem \ref{thm-calderjan18}, we have the following. 

\begin{corollary} 
Let $m \equiv 2 \pmod{4}$ and $2 \leq e \leq \lfloor m/2 \rfloor$. If $\gcd(m, e)=2$, then 
$\overline{\C_e}^\perp$ has parameters $[2^m, 2m+1, 2^{m-1}-2^{m/2}]$ and weight enumerator 
\begin{equation}\label{eqn-wtenumeratord}
\overline{A}^\perp(z)=1 + uz^{2^{m-1}-2^{m/2}} + vz^{2^{m-1}} + uz^{2^{m-1}+2^{m/2}} + z^{2^m},  
\end{equation} 
where 
\begin{eqnarray}\label{eqn-myuv}
u=(2^m-1)2^{m-2}, \ v=(2^m-1)(2^{m+1}-2^{m-1}+2). 
\end{eqnarray}
\end{corollary}    

\begin{theorem}\label{thm-desiredwtdist} 
Let $m \equiv 2 \pmod{4}$ and $2 \leq e \leq \lfloor m/2 \rfloor$. If $\gcd(m, e)=2$, then 
$\overline{\C_e}$ has parameters $[2^m, 2^m-1-2m, 4]$ and weight distribution  
\begin{eqnarray*}
2^{2m+1}\overline{A}_k   
& = &  (1+(-1)^k) \binom{2^m}{k} + 
\frac{1+(-1)^k}{2} (-1)^{\lfloor k/2 \rfloor} \binom{2^{m-1}}{\lfloor k/2 \rfloor} v  + \\ 
& & u \sum_{\substack{0 \le i \le 2^{m-1}-2^{m/2} \\
0\le j \le 2^{m-1}+2^{m/2} \\i+j=k}}[(-1)^i +(-1)^j] \binom{2^{m-1}-2^{m/2}} {i} \binom{2^{m-1}+2^{m/2}}{j}, 
\end{eqnarray*}
for $0 \le k \le 2^m$, where $u$ and $v$ are given in (\ref{eqn-myuv}). 
\end{theorem} 

\begin{proof} 
The parameters of $\overline{\C_e}$ were proved in Theorem \ref{thm-calderjan18}. 
The weight distribution formula for $\overline{\C_e}$ follows from  
the weight enumerator $\overline{A}^\perp(z)$ of $\overline{\C_e}^\perp$ in (\ref{eqn-wtenumeratord}) 
and Theorem \ref{thm-MI}. 
\end{proof} 

We are now ready to prove the main result of this paper. 

\begin{theorem}\label{thm-mmmainresult}
Let $m \equiv 2 \pmod{4}$, $2 \leq e \leq \lfloor m/2 \rfloor$, and $\gcd(m, e)=2$. Then 
the supports of the codewords of weight $4$ in $\overline{\C_e}$ form a $2$-$(2^m, 4, 1)$ 
design, i.e., a Steiner system $S(2, 4, 2^m)$. 
\end{theorem} 

\begin{proof}
Using the weight distribution formula $\overline{A}_k$ given in Theorem \ref{thm-desiredwtdist}, 
we obtain 
$$ 
\overline{A}_4= \frac{2^{m-1}(2^m-1)}{6}. 
$$ 
It then follows that 
$$ 
\lambda=\overline{A}_4 \frac{\binom{4}{2}}{\binom{2^m}{2}}=1. 
$$ 
This completes the proof. 
\end{proof} 

For every $m \equiv 2 \pmod{4}$ and $m \geq 6$, we can choose $e=2e_1$ with $\gcd(m/2, e_1)=1$ 
and $e_1 \leq \lfloor m \rfloor/2$. Such $e$ will satisfy the conditions in Theorem \ref{thm-mmmainresult}. At least we can choose $e=2$. This means that for every $m \equiv 2 \pmod{4}$ 
with $m \geq 6$, Theorem \ref{thm-mmmainresult} gives at least one Steiner system $S(2, 4, 2^m)$. 
In fact, it constructs more than one Steiner system $S(2, 4, 2^m)$. For example, when $m=14$, we 
can choose $e$ to be any element of $\{2, 4, 6\}$. 
Therefore, Theorem \ref{thm-mmmainresult} gives an infinite family of Steiner system $S(2, 4, 2^m)$.

In addition to the infinite family of Steiner systems $S(2, 4, 2^m)$, Theorem \ref{thm-mmmainresult} 
gives many other $2$-designs. Below we present two more examples. 

\begin{theorem}\label{thm-wt6design}
Let $m \equiv 2 \pmod{4}$, $2 \leq e \leq \lfloor m/2 \rfloor$, and $\gcd(m, e)=2$. Then 
the supports of the codewords of weight $6$ in $\overline{\C_e}$ form a $2$-$(2^m, 6, \lambda)$ 
design, where 
$$ 
\lambda=\frac{(2^m-4)(2^m-24)}{24}. 
$$ 
\end{theorem} 

\begin{proof}
Using the weight distribution formula $\overline{A}_k$ given in Theorem \ref{thm-desiredwtdist}, 
we obtain 
$$ 
\overline{A}_6= \frac{2^{m}(2^m-1)(2^m-4)(2^m-24)}{720}. 
$$ 
It then follows that 
$$ 
\lambda=\overline{A}_6 \frac{\binom{6}{2}}{\binom{2^m}{2}}=\frac{(2^m-4)(2^m-24)}{24}. 
$$ 
This completes the proof. 
\end{proof} 

\begin{theorem}\label{thm-wt8design}
Let $m \equiv 2 \pmod{4}$, $2 \leq e \leq \lfloor m/2 \rfloor$, and $\gcd(m, e)=2$. Then 
the supports of the codewords of weight $8$ in $\overline{\C_e}$ form a $2$-$(2^m, 8, \lambda)$ 
design, where 
$$ 
\lambda=\frac{(2^m-4)(2^{3m}-23\times 2^{2m}+344 \times 2^m -1612)}{720}. 
$$ 
\end{theorem} 

\begin{proof}
Using the weight distribution formula $\overline{A}_k$ given in Theorem \ref{thm-desiredwtdist}, 
we obtain 
$$ 
\overline{A}_8= \frac{2^{m}(2^m-1)(2^m-4)(2^{3m}-23\times 2^{2m}+344 \times 2^m -1612)}{2 \times 20160}. 
$$ 
It then follows that 
$$ 
\lambda=\overline{A}_8 \frac{\binom{8}{2}}{\binom{2^m}{2}}= 
\frac{(2^m-4)(2^{3m}-23\times 2^{2m}+344 \times 2^m -1612)}{720}. 
$$ 
This completes the proof. 
\end{proof} 

We point out that the main result in Theorem \ref{thm-mmmainresult} of this paper and 
Theorems \ref{thm-wt6design} and \ref{thm-wt8design} cannot be proved with the Assmus-Mattson 
Theorem due to the weight distribution of $\overline{\C_e}^\perp$ and the low minimum distance 
of $\overline{\C_e}$.  

When $m$ is odd and $\gcd(m, e)=1$, the code $\C_e$ and their relatives are also very interesting 
due to the following: 
\begin{itemize}
\item The code $\C_e$ and its dual $\C_e^\perp$ hold many infinite families of 
      $2$-designs. 
\item The extended code $\overline{\C_e}$ and its dual $\overline{\C_e}^\perp$ hold  
      many infinite families of $3$-designs.         
\end{itemize} 
These results were proved by the Assmus-Mattson Theorem, and the designs of those codes were 
covered in \cite{DingLi16}. 

When $m/\gcd(m,e)$ is even and $1 \leq e \leq \overline{m}$, one can find an algebraic 
expression of the weight distribution of the code $\overline{\C_e}$ with the weight distribution 
of $\overline{\C_e}^\perp$ depicted in Table \ref{tab-aaMinejan18} and Theorem \ref{thm-MI}, 
and then determine the parameters of some of the two designs held in $\overline{\C_e}$.

\subsection{Designs from some other codes $\C_E$ and their relatives} 

In Section \ref{sec-mainsect}, we treated the designs from the code $\C_{\{1+2^e\}}$ and its relatives. 
In this section, we provide information on designs from other codes $\C_E$ and their relatives. 

When $m \geq 5$ is odd and $E=\{(m-3)/2, (m-1)/2\}$ or $E=\{1, 2\}$, $\C_E$ has parameters 
$[2^m-1, 2^m-1-3m, 7]$ and $\overline{\C_E}$ has parameters $[2^m, 2^m-1-3m, 8]$. 
$\overline{\C_E}^\perp$ has dimension $3m+1$ and has six weights. In this case, $\C_E$ and 
$\C_E^\perp$ hold many infinite families of $2$-designs, while the codes $\overline{\C_E}$ and $\overline{\C_E}^\perp$ hold many infinite families of $3$-designs. These designed were treated in 
\cite{DingDesign16}.   

When $m \geq 4$ is even and $E=\{1, 2\}$, $\C_E$ does not hold $2$-designs. But $\overline{\C_E}$ 
and $\overline{\C_E}^\perp$ hold $2$-designs. The parameters of these $2$-designs were studied in 
\cite{DingZhouConf17}.  

When $m \geq 4$ is even and $E=\{(m-2)/2, m/2\}$, $\C_E$ has parameters $[2^m-1, 2^m-1-3m/2, 5]$, 
$\overline{\C_E}$ has parameters $[2^m, 2^m-1-3m/2, 6]$, and the weight distribution of 
$\overline{\C_E}^\perp$ is known \cite{Kasa69}. The parameters of the $2$-designs held in 
$\overline{\C_E}$ and $\overline{\C_E}^\perp$ are the same as those of the $2$-designs held 
in some codes in \cite{DingZhouConf17}. 

When $m \geq 7$ is odd and $E=\{(m-5)/2, (m-3)/2, (m-1)/2\}$, $\C_E^\perp$ has dimension $4m$ 
and has $7$ weights \cite{Kasa69}. It can be prove that $\C_E$ has parameters $[2^m-1, 2^m-1-4m, 7]$. 
The weight distribution of $\overline{\C_E}^\perp$ can be determined. Hence, the parameters of 
the $2$-designs held in $\overline{\C_E}^\perp$ and some of the $2$-designs held in 
$\overline{\C_E}$ can be worked out.

\section{Concluding remarks} 

While a lot of $t$-designs from codes have been constructed 
(see \cite{AK92, AM74, BJL, DingDesign16, DingLi16, KP95, Tonc93, Tonchev, Tonchevhb}, and the references therein), only a few constructions of infinite families of Steiner systems from codes are known in the literature. One of 
them is the Steiner quadruple systems $S(3, 4, 2^m)$ from the minimum codewords in the binary 
Reed-Muller codes $\cR_2(m-2, m)$. Another one is the Steiner triple systems $S(2, 3, 2^m-1)$ 
from the minimum codewords in the binary Hamming codes. This paper has now filled the gap of 
constructing an infinite family of Steiner systems $S(2, 4, v)$ from codes. We inform the reader that an infinite family of 
conjectured Steiner systems $S(2, 4, (3^m-1)/2)$ was presented in \cite{DingLi16}. It would be good if more infinite families of Steiner systems from error correcting codes could be discovered.

\end{document}